\documentclass[11pt,leqno]{article}
\usepackage{amsmath, amscd, amsthm, amssymb, graphics, xypic, mathrsfs, setspace, fancyhdr, times, bm, pdfsync, enumitem}
\usepackage[usenames, dvipsnames, svgnames, table]{xcolor}
\usepackage[colorlinks=true,pagebackref=true]{hyperref} 
\hypersetup{backref}

\newcommand{\greg}[1]{\textcolor{red}{#1}} 
\newcommand{\odile}[1]{\textcolor{green}{#1}} 

\newcommand{\N}{{\mathbb N}}
\newcommand{\Tr}{{\mathrm {Tr}}}

\newcommand{\rank}{{\mathrm {rank}}}
\newcommand{\wt}{{\mathrm {wt}}}
\newcommand{\maxwt}{{\mathrm {maxwt}}}
\newcommand{\Res}{{\mathrm {Res}}}

\newcommand{\Span}{{\rm Span}}
\newcommand{\Rsupp}{\mathrm{Rsupp}}
\newcommand{\Gal}{{\rm Gal}}

\newcounter{intro}
\theoremstyle{plain}
\newtheorem{thm}{Theorem}[section]

\newtheorem{lem}[thm]{Lemma}
\newtheorem{cor}[thm]{Corollary}
\newtheorem{prop}[thm]{Proposition}
\newtheorem*{claim*}{Claim} 

\newtheorem*{thm*}{Theorem}
\newtheorem*{problem*}{Problem}

\theoremstyle{definition}
\newtheorem{defn}[thm]{Definition}
\newtheorem*{defn*}{Definition}

\newtheorem{notation}[thm]{Notation}

\theoremstyle{remark}
\newtheorem{rem}[thm]{Remark}

\numberwithin{equation}{subsection}

\begin{document}

\author{Gr\'egory Berhuy\and Jean Fasel\and Odile Garotta}

\title{{\bf Rank weights for arbitrary finite field extensions }}
\date{}
\maketitle

\begin{abstract}
In this paper, we study several definitions of generalized rank weights for arbitrary finite extensions of fields. We prove that all these definitions coincide, generalizing known results for extensions of finite fields.
\end{abstract}

\begin{footnotesize}
\tableofcontents
\end{footnotesize}

\section{Introduction}
  
Rank metric codes were  introduced by Gabidulin  in \cite{Gabidulin85}  as linear codes over a  finite extension of a given finite field.
 Their application to network coding  has recently induced a lot of theoretical research, and
many notions in the theory for codes with the Hamming metric have  found an
equivalent  for codes with the rank metric. One of these is the definition 
of the {\sl generalized Hamming weights} associated to the Hamming distance of a code,  which
has given rise to several proposed definitions of  the \emph{
 generalized rank weights} of a  rank metric code.
 The first definition was given by Kurihara-Matsumoto-Uyematsu \cite{Kurihara15}, an alternative 
one was proposed  by Oggier-Sboui in \cite{Oggier12},  
  then Ducoat proved  in \cite[Prop II1]{Ducoat15} the  equivalence of the first definition with a refinement of the second.
In \cite{Jurrius17}, motivated by their interest in the \emph{$r$-th generalized rank weight enumerator} of a rank metric code, Jurrius-Pellikaan 
gave yet  another definition of these weights. Moreover,  they considered
 the general setting of finite  Galois extensions, a study initiated  by Augot-Loidreau-Robert in \cite{Augot13}.
They introduced the notion of rank support of a vector and of a code, proving for finite fields that, in case  the length $n$ of the code is bounded by the degree $m$ of the extension,
there exists a vector in the code $C$ which has same rank support as $C$ (\cite[Proposition 3.6]{Jurrius17} ). This was a
 crucial step in showing  the equivalence of all the proposed definitions in the case of  finite fields (\cite[5.2]{Jurrius17})
and \cite[Remark 3.7]{Jurrius17} raised the question of whether this statement 
 holds more generally. 
 To study the case of general  
 Galois extensions, 
Jurrius-Pellikaan used the field trace and the Galois closure of a code and  proved 
the equivalence of their definition with the one given by  Kurihara-Matsumoto-Uyematsu. 
In the restricted case of {cyclic} extensions, they  linked rank support and Galois closure and showed the equivalence 
of their definition  with the refined definition of \cite{Ducoat15},  
always under the hypothesis $n\leq m$.

In  this paper we positively answer the question of generalizing \cite[3.6]{Jurrius17} for \emph{any} finite extension
(Theorem \ref{thm:geneJuPe}).

\begin{thm*}
Let $L/k$ be a finite field extension of degree $m$ and let $C\subset L^n$ be a linear subspace. If $m\geq n$, then there exists ${\bf c}\in C$ such that 
\[
\mathrm{Rsupp}(C)=\mathrm{Rsupp}({\bf c}).
\]
\end{thm*}

Moreover, we generalize the notion of the Galois closure of an $L$-linear code to an arbitrary finite field extension (Definition \ref{C*}).

\begin{defn*}
Let $L/k$ be {\bf any} field extension, and let $n\geq 1.$ For any $L$-linear subspace $C$ of $L^n$, we denote by $C^*$  the intersection of all $L$-linear subspaces of $L^n$ 
extended from $k^n$ and containing $C$.
\end{defn*}

We prove that $C^*$ coincides with the Galois closure of $C$ in case the field extension is Galois, using that for such an extension a linear code $C$ in $L^n$ is Galois invariant if and only if it has a basis in $k^n$ (\cite{Giorgetti10}). This generalization of the Galois closure of an $L$-linear code allows to consider the following definitions (again, for an arbitrary extension of fields).

\begin{defn*}\label{4def-grw} 
\begin{enumerate}
\item $d_{R,r}(C)=\min\limits_{\stackrel{D\subset C} {\dim(D)=r}}\wt_R(D)$, following \cite[Definition 2.5]{Jurrius17} 
\item ${\cal M}_r(C)=\min\limits_{\stackrel{V\subset L^n,V=V^*} {\dim(C\cap V)\geq r}}\dim V$,   following  \cite[Definition 5]{Kurihara15}
\item   $O\!S_r(C)=\min\limits_{\stackrel{D\subset C} {\dim(D)=r}} \maxwt_R(D)$,  following \cite{Oggier12}
\item  $D_r(C)=\min\limits_{\stackrel{D\subset C} {\dim(D)=r}} \maxwt_R(D^*)$, following \cite{Ducoat15}. 
\end{enumerate}
\end{defn*}

In this setting, we finally present a unified proof of the equivalence of all proposed definitions of generalized rank weights (Theorem \ref{thm:geneJuPe}).

\begin{thm*}
Let $L/k$ be a finite 
extension of degree $m$, let $n\geq 1$ be an integer such that $n\leq m$ and let $C$ be an $L$-linear subspace of $L^n$. Then the above four definitions coincide. 
\end{thm*}

The paper is organized as follows. Section \ref{sec-basics}  presents basic results on rank supports and their relations with the restriction
of a code. In particular, Proposition \ref{Rsupp-perp} may be seen as a generalization of the well-known result
 of Delsarte \cite[Theorem 2]{Delsarte75}  to arbitrary finite extensions.
Section \ref{geom}   is devoted to giving a geometric proof of Theorem \ref{thm:geneJuPe}. We warn the reader that this section involves some algebraic geometry, and in particular the notion of Weil restriction. 
In section \ref{sec-C*}  we present our generalization of the notion of Galois closure of an $L$-linear code for arbitrary field extensions (Definition \ref{C*}). 
We show that our closure has nice properties, in particular we prove in  Proposition \ref{Rsupp-C*} 
that for finite extensions it is precisely the $L$-extension of the rank support, thus generalizing the Galois case
 \cite[5.5]{Jurrius17}.
In section \ref {equiv-def}, we  
 use  
this proposition and Theorem \ref{thm:geneJuPe}  to prove the 
equivalence of all four  (extended) definitions of the $r$-th generalized rank weights (under the hypothesis $m\geq n$), in the setting of {arbitrary}  finite extensions (Theorem \ref{equivdef}).

\subsection*{Acknowledgments}
The authors warmly thank Fr\'ed\'erique Oggier for a gentle introduction to rank metric codes and
to the paper \cite{Jurrius17}.

\section{Basics on rank supports}\label{sec-basics}

Let $k\subset L$ be a finite field extension and let $n\in\N^*$. Given an $L$-linear subspace $C$ of $L^n$, we define (following \cite{Giorgetti10}) its \emph{restriction} $\Res (C)$
as the $k$-linear subspace $C\cap k^n$ of  $k^n$. Given a $k$-linear subspace $D$  of $k^n$, we let $D_{L}$ denote the $L$-linear subspace of $L^n$
generated by $D$. Note that $D_L$ is canonically isomorphic to $D\otimes_k L$. If $\varphi: D\to D^\prime$ is a $k$-linear map, we denote by $\varphi_L$ the $L$-linear map $D_L\to D^\prime_L$.

Let $m=[ L:k]$ and let $
\alpha_1,\ldots,\alpha_m$ be a $k$-basis of $L$. For any $n$-tuple ${\bf c}=(c_1,\ldots,c_n)\in L^n$, we may write $c_j=\sum_{i=1}^mc_{ij}\alpha_i$ and obtain a matrix $M({\bf c}):=(c_{ij})\in M_{mn}(k)$. 

The definitions of rank supports of a vector and of a subspace of $L^n$  were introduced by Jurrius and Pellikaan \cite{Jurrius17}. 

\begin{defn}
\begin{itemize}
\item[a)] For any ${\bf c}\in L^n$, we call \emph{rank support} of ${\bf c}$ the subspace of $k^n$ generated by the rows of 
$M({\bf c})$ and write it $\mathrm{Rsupp}({\bf c})$.
Its dimension is the \emph{rank weight} of ${\bf c}$, denoted by $\wt_R({\bf c})$.
\item[b)] Let  $C\subset L^n$  be an $L$-linear subspace. We call \emph{rank support} of $C$ the $k$-linear subspace 
$$\Span_k (\mathrm{Rsupp}({\bf c})\vert \,\break{\bf c}\in C)$$ 
and write it $\mathrm{Rsupp}(C)$. Its dimension is the \emph{rank weight} of $C$, denoted by $\wt_R(C)$.
\item[c)] The \emph{rank distance} $d_R(C)$ of an $L$-linear subspace $C$ of $L^n$ is the minimum rank of a matrix 
$M({\bf c})$, for ${\bf c}$ nonzero in $C$. That is, we have  $d_R(C)=
\min\limits_{{\bf c}\in C\setminus \{0\}} \wt_R({\bf c})$.
\end{itemize}
\end{defn}

We now  state some  basic properties of rank supports. 

\begin{prop}\label{prop:rkbasic}
Let ${\bf c}$, ${\bf c}'\in L^n$, let  $\alpha\in L^*$ and let $C$ be an  $L$-linear subspace of $L^n$.
 Then  
\begin{enumerate}
\item $\mathrm{Rsupp}({\bf c})$ does not depend on the choice of the $k$-basis 
of $L$.
\item $\mathrm{Rsupp}(\alpha {\bf c})= \mathrm{Rsupp}({\bf c})$.
\item $\mathrm{Rsupp}({\bf c}+{\bf c}')\subset  \mathrm{Rsupp}({\bf c})+ \mathrm{Rsupp}({\bf c}')$.
\item If ${\bf c}_1,\ldots,{\bf c}_r$ generate $C$ as an $L$-linear space, then  $\mathrm{Rsupp}(C)$ is the $k$-linear sum
of the \break $\mathrm{Rsupp}({\bf c}_i)$, $i=1,\ldots,r$.
\item Suppose ${\bf c}\neq 0$. Then  $\wt_R({\bf c})=1$  if and only if 
there exists $\lambda\in L^*$ such that $\lambda{\bf c}\in k^n$ . In that case we
have $\,\mathrm{Rsupp}(\lambda{\bf c})=k\lambda{\bf c}$.
\item $C$ has rank distance $d_R(C)=1$ if and only if $\,\Res (C)\neq \{0\}$.
\item 
We have the inclusion of $k$-subspaces $\Res (C)\subset \mathrm{Rsupp}(C)$.
\item 
We have the inclusion of $L$-subspaces $C \subset \mathrm{Rsupp}(C)_L\,$.
 \end{enumerate}
\end{prop}

\begin{proof} 
Properties 1-4  are proved in \cite[\S 2.3]{Jurrius17}. 
 Let us prove statement 5 :
we may assume that the first vector $\alpha_1$ of our $k$-basis of $L$
is $1$.  Then for any  nonzero ${\bf c}'\in k^n$, all rows of $M({\bf c}')$ are zero except the first one, which
is the $n$-tuple ${\bf c}'$ itself. We then get  $\Rsupp({\bf c}')=k{\bf c}'$ and $\wt_R({\bf c}')=1$. Thanks to
statement 2, it remains to consider the situation when $M({\bf c})$ has rank one. In that case there exists
an $l$ such that all $m$ rows of $M({\bf c})$ are colinear with the nonzero $l$-th row : ${\bf c}^{(i)}=\mu_i
{\bf c}^{(l)}$,  $\mu_i \in k$, $1\leq i\leq m$. We get ${\bf c}=(\sum_{i=1}^m\mu_i\alpha_i){\bf c}^{(l)}\neq 0$, that is $\lambda{\bf c}
\in k^n$  with $\lambda$ the inverse of  $\sum_{i=1}^m\mu_i\alpha_i$. So property 5 is proved.
Since the rank distance of $C$ is the minimum rank weight of a nonzero vector in $C$, 
it immediately yields statement 6.
We turn to  statement 7.
If ${\bf c}\in \Res(C)$, we know from  property 5 that $\Rsupp({\bf c})=k{\bf c}$. This gives the inclusion $\Res (C)\subset\Rsupp(C)$. 
Finally we  prove   statement 8. Using property 4, it suffices to consider the case where $C= L{\bf c}$. So
we need to prove that ${\bf c}\in \mathrm{Rsupp}({\bf c})_L\,$. But this is clear using the definition of
$\mathrm{Rsupp}({\bf c})$.
\end{proof}

\begin{prop}\label{k-bas-Res-Rsupp} 
Let $L/k$ be a finite extension, let $n\geq 1$, 
and let  $C\subset L^n$  be an $L$-linear subspace.  Then 
we have  $\Res (C)= \mathrm{Rsupp}(C)$ if and only if
 $C$ has a basis in $k^n$  or, equivalently, if and only if 
  $\dim _k\Res (C)= \dim C$.    
In particular, we then have $\mathrm{Rsupp}(C)_L=C$.
\end{prop} 

\begin{proof} 
Suppose that  $C$ has a basis
${\bf e}_1,\ldots, {\bf e}_r$  in $k^n$. 
Statement 5 of Proposition \ref{prop:rkbasic} shows that\break
 $\Rsupp({\bf e}_i)=k{\bf e}_i$  ($1\leq i\leq r$).
Using statement 4~of the same proposition we see that  the ${\bf e}_i$'s  generate $\Rsupp(C)$.
As ${\bf e}_i\in \Res(C)$, we get $\Rsupp(C)\subset \Res(C)$. We conclude that $\Res (C)= \mathrm{Rsupp}(C)$ using statement 7 in the previous proposition.

Suppose  conversely that $\Res (C)= \mathrm{Rsupp}(C)$ and $C\neq 0$. Then $\Res(C)\neq 0$. Let 
${\bf e}_1,\ldots, {\bf e}_s$
be a basis of $\Res(C)$. These vectors considered in $L^n$ are also independent over $L$.
We show that they  generate $C$ over $L$. Let ${\bf c}\in C$.
Then  $\Rsupp(C)$ contains 
 $\Rsupp({\bf c})$,  so our hypothesis
implies that $\Rsupp({\bf c})\subset \Res(C)$. Using a $k$-basis of the field $L$, we can thus 
write ${\bf c}$ as an $L$-linear combination of vectors in $\Res(C)$, therefore also  of the vectors 
${\bf e}_1,\ldots, {\bf e}_s$. Thus ${\bf e}_1,\ldots, {\bf e}_s$ is a basis of $C$ which  
 is  in $k^n$. Clearly this is equivalent to having $s=\dim C$.

The last equality then follows using the fact that $\Res (C)_L=C$ if $C$ has a basis in $k^n$.
\end{proof} 

We have the following general result. 

\begin{prop}\label{Thm-cas r<m} 
Let $L/k$ be a  finite extension of degree $m$, and let  $n\in\N^*$.  Let $C$ be an  $L$-linear subspace
 of $L^n$ of dimension $
\leq m$ and suppose that  $C$ admits an $L$-basis in $k^n$. Then 
 there exists 
${\bf c}\in C$ such that 
$\mathrm{Rsupp}(C)=\mathrm{Rsupp}({\bf c})$.
\end{prop}

\begin{proof}
 Let  ${\bf e}_1,\ldots,{\bf e}_r$ denote an $L$- basis of $C$ in $k^n$, with $r\leq m$. 
From Proposition \ref{k-bas-Res-Rsupp} we know that $\mathrm{Rsupp}(C)=\Res(C)$ 
is spanned by  ${\bf e}_1,\ldots,{\bf e}_r$.
 Let $\alpha_1,\ldots,\alpha_m$ be a $k$-basis of $L$ and set ${\bf c}=\sum_{i=1}^r \alpha_i\,{\bf e}_i$. Then
 ${\bf c}\in C$ and  the matrix $M({\bf c})$ has the $n$-tuples ${\bf e}_1,\ldots, {\bf e}_r$ as its first $r$ rows, 
and its $n-r$ remaining rows are zero. So $\Rsupp({\bf c})$ 
is spanned by  ${\bf e}_1,\ldots,{\bf e}_r$ and therefore equal to $\mathrm{Rsupp}(C)$.
\end{proof}

\begin{rem}\label{CN-extd}
Note that since $\wt_R({\bf c})\leq m$ for all ${\bf c}$
in $L^n$, the hypothesis $\dim C\leq m$ in Proposition \ref{Thm-cas r<m} 
is necessary.   
Indeed, if  $C$ admits  a 
basis in $k^n$, then $\Rsupp(C)$ has dimension $\dim C$ by \ref{k-bas-Res-Rsupp}; see also Remark
\ref{CNdim} for a stronger statement. 
\end{rem}

Recall that for any finite field extension $k\subset L$, we have the (field) trace $\Tr\colon L\to k$, which is $k$-linear and non trivial in case $L/k$ is separable. We can extend this map to a {\sl trace map} $\,\Tr\colon L^n\to k^n$, which is the component-wise extension of the field trace.  

\begin{defn} 
Let $k\subset L$ be a finite Galois extension with Galois group $G$.
The {\sl Galois closure} of an $L$-linear subspace $C$ of $L^n$, denoted by $C^*$, is the smallest subspace 
of $L^n$ that contains $C$ and is invariant under the component-wise action of $G$ on $L^n$. The subspace $C$
 is called {\sl Galois closed} in case $C$ is $G$-invariant, or equivalently in case $C=C^*$.
\end{defn}

The following result (\cite[Theorem 1]{Giorgetti10}) will play a crucial role in the sequel, notably in \ref{cor-c*c'} 
when we  set our generalized definition of the closure $C ^*$ of $C$.
 
\begin{thm}\label{thm :kbas}
If $L/k$ be a finite Galois extension 
and $C$ is an $L$-linear subspace of $L^n,$ then $C=C^*$ if and only if $C=\Res(C)_L$ or, equivalently, if and only if  $C$ admits a basis in $k^n$.

\end{thm}

For Galois extensions, Jurrius-Pellikaan have shown
 in \cite[Theorem 4.3]{Jurrius17} that the rank support of a subspace of $L^n$ is equal to its image by $\Tr$. This
 actually holds (with same proof)  for all {\sl separable} finite extensions $L/k$.

\begin{thm}\label{thm:Rsupp-Tr} 
 Suppose the finite  field extension $L/k$ is separable and let $n\geq 1$. Then for any $L$-linear subspace $C$ of $L^n$,
we have $\mathrm{Rsupp}(C)=\Tr(C).$
\end{thm}

Using this theorem we can 
 generalize the setting of  Theorem 9 in  \cite{Giorgetti10},  which   Giorgetti-Previtali have stated for
  Galois extensions (see \ref{thm :kbas}). 

\begin{prop}
Let $L/k$ be a finite separable extension and let $C$ be an $L$-linear subspace of $L^n$ for some $n\geq 1$. Then 
we have $\Res(C)=\Tr(C)  
$ if and only if  
$C$ admits a basis in $k^n$.
\end{prop}

\begin{proof}
This follows from  Proposition \ref{k-bas-Res-Rsupp}  combined with 
Theorem \ref{thm:Rsupp-Tr}.
\end{proof}

We denote by $\langle \_\,,\_ \rangle$ the standard bilinear form on $L^n$. Given any $L$-linear subspace $C$ of $L^n$,
we define its orthogonal as the linear subspace 
$$C^\perp =\{{\bf d}\in L^n \,|\, \forall  {\bf c}\in C, \, \langle{\bf c},{\bf d}\rangle\,=0\}\,.$$

\begin{lem}\label{perp-basis} 
Let $L/k$ be any field extension and let $C$ be an $L$-linear subspace of $L^n$ for some $n\geq 1$. If $C$ admits a basis in $k^n$, then so does $C^\perp$.
\end{lem}

\begin{proof}  Let  ${\bf e}_1,\ldots,{\bf e}_r$ denote an $L$- basis of $C$ in $k^n$ and let 
 $D$ be the orthogonal of $\Res(C)$ relative to the standard bilinear form on $k^n$. Then $C^\perp$
is the intersection of the $L{\bf e}_i^\perp$ ($1\leq i\leq r$), and $D$ has the same description over $k$. Therefore 
$C^\perp\supset D_L$.
Since dimensions are the same we get $C^\perp=D_L$, so  $C^\perp$ comes from $k^n$.
\end{proof}

For finite separable extensions,  we know from Delsarte  \cite[Theorem 2]{Delsarte75}  that 
$\Res (C)^\perp=\Tr(C^\perp)$, for any $L$-linear subspace $C$ of $L^n.$ 
Using Theorem \ref{thm:Rsupp-Tr}, the following relation can be seen as a generalization of that equality to arbitrary finite extensions.

\begin{prop}\label{Rsupp-perp}
Let $L/k$ be a finite extension, let $n\geq 1$, and
let $C$ be an $L$-linear subspace of $L^n.$ 
Then  we have $$\Res (C)^\perp=\Rsupp(C^\perp).$$
\end{prop}

\begin{proof}
We first show that $\Rsupp(C^\perp)\subset \Res (C)^\perp$.
Take ${\bf d}\in C^\perp$ and ${\bf c}\in \Res( C)$. We  have ${\bf d}=\sum_{i=1}^m \alpha_i\,{\bf d}^{(i)}$, where
$\alpha_1,\ldots,\alpha_m$ is a $k$-basis of $L$ and the ${\bf d}^{(i)}$'s belong to $\Rsupp(C^\perp)$. We need
to prove that all $\langle {\bf d}^{(i)},{\bf c}\rangle$ are zero. 
Indeed we have $0=\,\langle{\bf d},{\bf c}\rangle\,=\sum_{j=1}^n\bigl(\sum_{i=1}^m\alpha_i{d_j}^{(i)}c_j\bigr)=
\sum_{i=1}^m\alpha_i\langle{\bf d}^{(i)},{\bf c}\rangle$.

It now suffices to prove that $\Rsupp(C^\perp)^\perp\subset \Res (C)$. Let ${\bf c}\in k^n$ belong to
$\Rsupp(C^\perp)^\perp$ and consider ${\bf d}\in  C^\perp$. As above, write ${\bf d}=\sum_{i=1}^m \alpha_i\,{\bf d}^{(i)}$, with all ${\bf d}^{(i)}$'s in $\Rsupp(C^\perp)$.
Then we get $\langle{\bf c},{\bf d}\rangle=\sum_{i=1}^m\alpha_i \langle{\bf c},{\bf d}^{(i)}\rangle=0$. Therefore ${\bf c}\in (C^\perp)
^\perp=C$. Since ${\bf c}\in k^n$ the proof is complete.
\end{proof}

\begin{rem}\label{Rsupp-Res-perp}
In particular, the above proposition allows to compute the rank support of $C$. Indeed, as $(C^\perp)
^\perp=C$, we get 
\[
\Rsupp(C)=\Res(C^\perp)^\perp.
\]
\end{rem}

Finally we recover and complete the characterization of {\sl rank-degenerate} codes given by Jurrius-Pellikaan in 
\cite[Corollary 6.3]{Jurrius17}. Recall that these are the $L$-linear subspaces of $L^n$ such that $d_R(C^\perp)=1$.

\begin{cor}\label{carac-nondeg}
Let $L/k$ be a finite extension and let $n\geq 1$. A linear code $C$ in $L^n$ is degenerate with respect to the rank metric if and only
if we have $\Rsupp(C)\neq k^n$, or equivalently, if and only if $\Res(C^\perp)\neq \{0\}$.
\end{cor}

\begin{proof}
We apply  statement 6 of  Proposition \ref{prop:rkbasic} and combine it with 
 Remark \ref{Rsupp-Res-perp}.
\end{proof}

\section{The geometric viewpoint}\label{geom}

In this section, we prove the following theorem, which is our first main result.

\begin{thm}\label{thm:geneJuPe}
Let $L/k$ be a finite field extension of degree $m$ and let $C\subset L^n$ be a linear subspace. If $m\geq n$, then there exists ${\bf c}\in C$ such that 
\[
\mathrm{Rsupp}(C)=\mathrm{Rsupp}({\bf c}).
\]
\end{thm}

We start with a few observations allowing to reduce the problem to an algebro-geometric problem. First, we note that we can suppose that $\mathrm{Rsupp}(C)=k^n$. Indeed, by Proposition \ref{prop:rkbasic}, we have $C\subset \Rsupp(C)_L$. We may therefore choose a basis of $\Rsupp(C)$ and work with this vector space instead of $k^n$.

\begin{lem}\label{lem:supp}
Let $D\subset L^n$ be a $L$-linear subspace, and let $\varphi:k^n\to k$ be a linear map. Then 
\[
\varphi(\Rsupp(D))_L=\varphi_L(D).
\]
\end{lem}

\begin{proof}
As $D\subset \Rsupp(D)_L$, we have $\varphi_L(D)\subset \varphi(\mathrm{Rsupp}(D))_L$. Now, both terms are $L$-linear subspaces of $L$ and it suffices to prove that if $\varphi_L(D)=0$, then $\varphi(\mathrm{Rsupp}(D))_L=0$. Suppose then that $\varphi_L(D)=0$ and let ${\bf d}\in D$. We can write as before
\[
{\bf d}=\sum_{i=1}^m \alpha_i{\bf d}^{(i)}\in L^n,
\]
with ${\bf d}^{(i)}\in \mathrm{Rsupp}(D)$. By definition, 
\[
\varphi_L({\bf d})=\sum_{i=1}^m \alpha_i \varphi({\bf d}^{(i)})
\]
and linear independence implies that $\varphi({\bf d}^{(i)})=0$ for any $i=1,\ldots,m$. As $\mathrm{Rsupp}(D)$ is generated by the ${\bf d}^{(i)}$ for ${\bf d}\in D$, it follows that $\varphi(\mathrm{Rsupp}(D))_L=0$.
\end{proof}

\begin{rem}\label{rem:zerozero}
As a corollary, we see that $\varphi(\mathrm{Rsupp}(D))=0$ if and only if $\varphi_L(D)=0$, i.e. that 
\[
\mathrm{Rsupp}(D)\subset \ker \varphi \iff D\subset \ker \varphi_L.
\]
In case $\Rsupp(D)=k^n$, this yields that $\varphi=0$ if and only if $\varphi_L(D)=0$.
\end{rem}

\begin{prop}\label{prop:reformulation}
Let $L/k$ be a finite field extension of dimension $m$ and let $n\leq m$. Let $C\subset L^n$ be a plane with $\mathrm{Rsupp}(C)=k^n$. Then there exists ${\bf v}\in C$ such that  
\[
\varphi_L({\bf v})=0 \iff \varphi=0
\]
for any $k$-linear map $\varphi:k^n\to k$.
\end{prop}

The proof of Proposition \ref{prop:reformulation} will be of geometric nature. Before starting with it, we first recall a few facts about the Weil restriction of quasi-projective schemes. If $X$ is a quasi-projective scheme over $L$, then it Weil restriction  $\mathcal{R}_{L/k}X$ along the field extension $L/k$ is the quasi-projective scheme over $k$ satisfying 
\[
\mathrm{Hom}_k(\mathrm{Spec}(R),\mathcal{R}_{L/k}X)=\mathrm{Hom}_L(\mathrm{Spec}(R_L),X).
\]
for any $k$-algebra $R$ (its existence is given by \cite[\S 7.6]{Bosch90} and \cite[Proposition A.5.8]{Conrad15}).
For instance, if $X=\mathbb{A}^1_L$, then $\mathcal{R}_{L/k}\mathbb{A}^1_L\simeq \mathbb{A}[\mathrm{Hom}_k(L,k)]\simeq \mathbb{A}_k^m$ as a direct computation shows (or \cite[\S 4.6]{Poonen17}). In particular, we see that the Krull dimension of $\mathcal{R}_{L/k}\mathbb{A}^1_L$ is $m$. If $X=\mathbb{P}_L^1$, then \cite[Proposition A.5.9]{Conrad15} shows that $\mathcal{R}_{L/k}\mathbb{P}_L^1$ is connected and \cite[Proposition A.5.2]{Conrad15} implies that an open immersion $\mathbb{A}^{1}_L\subset \mathbb{P}_L^1$ induces an open immersion $\mathcal{R}_{L/k}\mathbb{A}^{1}_L\subset \mathcal{R}_{L/k}\mathbb{P}_L^1$. This shows that $\mathcal{R}_{L/k}\mathbb{P}_L^1$ is of dimension $m$. More generally, one can prove that the Weil restriction $\mathcal{R}_{L/k}X$ of a smooth quasi-projective scheme $X$ of (pure) dimension $d$ is of dimension $dm$ but we don't use this fact here.

The main ingredient of the proof of Proposition \ref{prop:reformulation} is a morphism of schemes
\[
f:\mathbb{P}_k^{n-1}\to \mathcal{R}_{L/k}\mathbb{P}(C)\simeq \mathcal{R}_{L/k}\mathbb{P}^1_L
\]
that we now define. Let $R$ be a finitely generated $k$-algebra. Recall from \cite[Chapitre 4]{EGAII} that the set of $R$ points of $\mathbb{P}^{n-1}_k$ is identified with a line bundle $M$ over $R$ and an equivalence class of surjections 
\[
\varphi:R^n\to M
\]
under the natural action of $R^\times$. To such a surjection, we can associate its kernel $\ker \varphi$, which is a projective $R$-module $P$ of rank $n-1$. Extending scalars to $L$, we obtain a projective $R_L$-module $P_L$ which is a submodule of $R^n_L$. We can then consider the intersection $P_L\cap (C\otimes_L R_L)\subset C\otimes_L R_L$ where $C$ is the plane in the problem.

\begin{lem}
The submodule $P_L\cap (C\otimes_L R_L)$ is locally free of rank $1$.
\end{lem}

\begin{proof}
It suffices to prove that for any maximal ideal $\mathfrak m$ of $R_L$, the module is free of rank one after tensoring with $L(\mathfrak m)=R_L/\mathfrak m$. Now, $L(\mathfrak m)$ is a finite field extension of $L$, and thus of $k$. By construction, $P_{L(\mathfrak m)}=\ker \varphi_{L(\mathfrak m)}$ and we have to compute the dimension of $P_{L(\mathfrak m)}\cap C_{L(\mathfrak m)}$. As $P_{L(\mathfrak m)}$ is of dimension $n-1$ and $C_{L(\mathfrak m)}$ of dimension $2$, the intersection can't be trivial. Suppose that $P_{L(\mathfrak m)}\cap C_{L(\mathfrak m)}=C_{L(\mathfrak m)}$, i.e. that $C_{L(\mathfrak m)}\subset P_{L(\mathfrak m)}$. This means that $\varphi_{L(\mathfrak m)}(C_{L(\mathfrak m)})=0$ and then that $\varphi_L(C)=0$. By Remark \ref{rem:zerozero} this implies that $\varphi(\mathrm{Rsupp}(C))=0$, i.e. that $\varphi=0$, which is excluded. Thus, $P_L\cap (C\otimes_L R_L)$ is locally free of rank $1$.
\end{proof}

The lemma shows that the map
\[
f:\mathbb{P}_k^{n-1}\to \mathcal{R}_{L/k}\mathbb{P}(C)
\]
is well-defined. To make it more explicit, we compute the image of a $k$-point $[a_1:\ldots:a_n]$ of $\mathbb{P}_k^{n-1}$. A representative of the equivalence class of linear maps $\varphi:k^n\to k$ associated to this point is
\[
{\bf a}:k^n\to k
\]
given by $(u_1,\ldots,u_n)\mapsto \sum_{i=1}^na_iu_i$. If $C=\Span ({\bf v},{\bf w})\subset L^n$ with  ${\bf v}=(v_1,\ldots,v_n)$ and ${\bf w}=(w_1,\ldots,w_n)$ is such that $\mathrm{RSupp}(C)=k^n$, we see that an element of the form $\alpha {\bf v}+\mu {\bf w}$ is in the kernel of $a_L$ 
if and only if $\alpha(\sum a_iv_i)+\beta (\sum a_iw_i)=0$. It follows that 
\[
f([a_1:\ldots:a_n])=[\sum a_iw_i:-\sum a_iv_i].
\]

We are now ready to prove Proposition \ref{prop:reformulation}. We observe first that a rational point of $\mathcal{R}_{L/k}\mathbb{P}(C)$ corresponds to a $L$-point of $\mathbb P(C)$, i.e. to a $L$-linear subspace $L\cdot {\bf v}$ of $C$, where ${\bf v}\in C$ is non trivial. This rational point is in the image if and only if there exists a surjection $\varphi:k^n\to k$ whose kernel $P$ has the property that $P_L\cap C=L\cdot {\bf v}$, i.e. $\varphi_L({\bf v})=0$ and $\varphi_L(C)\neq 0$. In view of Lemma \ref{lem:supp}, we see that a rational point $L\cdot {\bf v}\in\mathcal{R}_{L/k}\mathbb{P}(C)$ is in the image of $f$ if and only if there exists a surjection $\varphi:k^n\to k$ such that $\varphi_L({\bf v})=0$. Thus, we see that a rational point of $\mathcal{R}_{L/k}\mathbb{P}(C)$ not in the image of $f$ satisfies the conclusion of Proposition \ref{prop:reformulation} and it suffices to show that the morphism $f$ is not surjective on $k$-points. Let $Z\subset \mathcal{R}_{L/k}\mathbb{P}(C)$ be the closure of the image of $f$. It is of dimension at most the dimension of $\mathbb{P}_k^{n-1}$, i.e. at most $n-1$, while $\mathcal{R}_{L/k}\mathbb{P}(C)$ is of dimension $m=[L:k]$. As $m\geq n$, we see that $Z$ is a proper closed subset of $\mathcal{R}_{L/k}\mathbb{P}(C)$. If $k$ is infinite, it follows that there is a rational point in the open complement of $Z$ and the result is proved. If $k$ is finite, we observe that $\mathbb{P}^{n-1}_k$ has $\frac {\vert k\vert^n-1}{\vert k\vert-1}$ rational points. On the other hand, the ($k$-)rational points of $\mathcal{R}_{L/k}\mathbb{P}(C)$ are in bijection with the ($L$)-rational points of $\mathbb{P}^1_L$ and it follows that there are $\frac {\vert L\vert^2-1}{\vert L\vert-1}$ such points.  We conclude using the fact that  
\[
\frac {\vert k\vert^n-1}{\vert k\vert-1}<\frac {\vert L\vert^2-1}{\vert L\vert-1}.
\]
\qed

We now proceed with the proof of theorem \ref{thm:geneJuPe}.

\begin{proof}[proof of Theorem \ref{thm:geneJuPe}]
If $C$ is of dimension $\leq 1$, then there is nothing to do. Suppose then that $n\geq 2$ and that $C\subset L^n$ is a plane. As usual, we may assume that $\mathrm{Rsupp}(C)=k^n$. By Proposition \ref{prop:reformulation}, there exists ${\bf c}\in C$ such that ${\bf c}\in \ker(\varphi_L)$ if and only if $\varphi=0$. If $\mathrm{Rsupp}({\bf c})\neq k^n$, there exists a non trivial linear form $\psi:k^n\to k$ such that $\psi(\mathrm{Rsupp}({\bf c}))=0$. By Lemma \ref{lem:supp}, this forces $\psi({\bf c})=0$ which contradicts the assumption on ${\bf c}$.  
So the theorem 
holds in this case.

To prove the 
general case, we now work by induction on the dimension $2\leq r\leq n$ of $C$. We can write $C=D\oplus L{\bf e}$ for some $D$ of dimension $r-1$ and some  nontrivial ${\bf e}$. By induction, there exists ${\bf d}\in D$ such that $\mathrm{Rsupp}({\bf d})=\mathrm{Rsupp}(D)$. Now, $\mathrm{Rsupp}(C)$ is spanned by $\mathrm{RSupp}(D)$ and $\mathrm{Rsupp}({\bf e})$, i.e. by  $\mathrm{Rsupp}({\bf d})$ and $\mathrm{Rsupp}({\bf e})$. 
By Proposition \ref{prop:rkbasic} 
we are thus reduced to considering $C^\prime=\Span_L({\bf d},{\bf e})$ and we can conclude using the case of a plane proved above.
\end{proof}

\begin{rem}\label{<=m nec} 
In particular, taking $C=L^n$ shows that the condition $m\geq n$ is necessary
in order for such an element  ${\bf c}\in C$  to exist  {\bf for any} subspace $C$ of $L^n$ 
(indeed $M({\bf c})$ have $m$ rows); see also 
Remark \ref{CNdim}. 
\end{rem}

Finally, recall from Corollary \ref{carac-nondeg} 
that  a code $C$ in $L^n$ is \emph{nondegenerate}  with 
respect to the rank metric  if and only if  
$\mathrm{Rsupp}(C)=k^n$. 
 Theorem \ref{thm:geneJuPe} thus yields the following equivalence :

\begin{prop}\label{thm nondeg}
 Let $L/k$ be an extension of degree $m$, and  let $C$ be  an $L$-linear code of length $n$ which is nondegenerate
with respect to the rank metric. Then there is a ${\bf c}\in C$ such that $\mathrm{Rsupp}(C)=\mathrm{Rsupp}({\bf c})$
if and only if $m\geq n$.
\end{prop}

\begin{proof}
Since  matrices $M({\bf c})$ have $m$ rows, the equality  $\mathrm{Rsupp}({\bf c})=k^n$ forces $m\geq n$.
The converse is a special case of  Theorem \ref{thm:geneJuPe}. 
\end{proof}


\section{A generalization of the Galois closure of a subspace}\label{sec-C*}

\begin{defn}
Let $L/k$ by {\bf any} field extension, and let $n\geq 1.$ We say that an $L$-linear subspace $C$ of $L^n$ {\it is extended from $k^n$} if $C$ is spanned by a family of vectors in $k^n$.
Equivalently, $C$ has a basis in $k^n$, or else we have $C=\Res(C)_L$.
\end{defn}

\begin{lem}\label{lem-interckn}
Let $L/k$ be any field extension, and let $n\geq 1$. If $(C_i)_{i\in I}$ is a family of subspaces of $L^n$ extended from $k^n$, then $\displaystyle\bigcap_{i\in I}C_i$ is extended from $k^n$.
\end{lem}

\begin{proof}  
For each $i\in I$, one may write $C_i=(D_i)_L^\perp$, for some $k$-linear subspace $D_i$ of $k^n$, where the orthogonal is taken with respect to the standard $L$-bilinear form on $L^n.$
Indeed, $C_i$ is the orthogonal of $C_i^\perp$, which by Lemma \ref{perp-basis} is extended from $k^n$.

Therefore, we have $$\bigcap_{i\in I}C_i=\bigcap_{i\in I}(D_i)_L^\perp=\Span_L((D_i)_L, i\in I)^\perp=(\Span_
k(D_i, i\in I)^\perp)_L\,.$$ Thus, $\bigcap_{i\in I}C_i$ is extended from $k^n$.
\end{proof}

Note that, given an $L$-linear subspace $C$ of $L^n$, there is at least one $L$-linear subspace of $L^n$ extended from $k^n$ and containing $C$, namely $L^n$ itself. Therefore, in view of lemma  \ref{lem-interckn}, the following statement makes sense.

\begin{cor}\label{cor-c*c'}
Let $L/k$ be a finite Galois extension with Galois group $G$, let $n\geq 1$, and let $C$ be an $L$-linear subspace of $L^n.$ Then $C^*$ is the intersection of all $L$-linear subspaces of $L^n$ extended from $k^n$ and containing $C$.
\end{cor}

\begin{proof}
Let $C'$ be the intersection of all $L$-linear subspaces of $L^n$ extended from $k^n$ and containing $C$.

Since $C^*$ is an $L$-linear subspace of $L^n$ extended from $k^n$ (by Theorem \ref{thm :kbas}) and containing $C$, we have $C'\subset C^*$. Conversely, we have $C\subset C'$ by definition. Since $C'$ is extended from $k^n$ (by Lemma \ref{lem-interckn}), Theorem \ref{thm :kbas} shows that   it 
is $G$-invariant. By definition of $C^*$, we get $C^*\subset C'$.
\end{proof}

This corollary motivates the following generalization of the definition of $C^*$.

\begin{defn}\label{C*}
Let $L/k$ be {\bf any} field extension, and let $n\geq 1.$ For any $L$-linear subspace $C$ of $L^n$, we denote by $C^*$  the intersection of all $L$-linear subspaces of $L^n$ 
extended from $k^n$ and containing $C$.

It is the smallest $L$-linear subspace of $L^n$ extended from $k^n$ and containing $C$.
\end{defn}

\begin{rem}\label{C*Galois=C*}

\begin{enumerate}
\item If $L/k$ is a finite Galois extension, we recover the usual definition of $C^*$, by Corollary \ref{cor-c*c'}. However, this generalization is more convenient, 
since its does not need the extension $L/k$ to be Galois.

\item The definition of $C^*$  immediately yields the following properties.
\begin{enumerate}

\item for all $L$-linear subspaces $C$ of $L^n$, we have $C\subset C^*$.

\item for all $L$-linear subspaces $C$ of $L^n$, we have $C=C^*$ if and only if $C$ is extended from $k^n$. In particular, $(C^*)^*=C^*$.

\item for all $L$-linear subspaces $D$ of $L^n$ extended from $k^n$ such that $C\subset D$, we have $C^*\subset D$.

\item for all $L$-linear subspaces $C,D$ of $L^n$ such that $C\subset D$, we have $C^*\subset D^*$.
\end{enumerate} 
\end{enumerate}
\end{rem}

\begin{prop}\label{sum-C*}
Let $L/k$ be any field 
extension and let $C$, $C'$ be  $L$-linear subspaces of $L^n$ for some $n\geq 1$. Then 
we have $(C+C')^*=C^*+C'^{\,*}$.
\end{prop}

\begin{proof}
By Remark \ref{C*Galois=C*} the subspace $(C+C')^*$ contains both
$C^*$ and $C'^*$, thus it contains their sum.
Now combining Lemmas \ref{perp-basis} and \ref{lem-interckn}, we see that given two subspaces
$D$, $D'$ of $L^n$ which are extended from $k^n$, their sum also is extended from $k^n$. Indeed, the orthogonal
of $D+D'$ is the intersection of $D^\perp$ and $D'^{\,\perp}$. Therefore $C^*+C'^{\,*}$ is extended from $k^n$.
Since it contains $C+C'$, it also contains $(C+C')^*$.
\end{proof}

Our next proposition  generalizes Proposition 5.5 and Lemma 5.6 in \cite{Jurrius17}. In particular, it then follows 
 from the
first item in Remark \ref{C*Galois=C*}  that 
 their  statements 5.5, 5.6 and 5.7 hold for {\sl any}  Galois extension, not necessarily cyclic.  

\begin{prop}\label{Rsupp-C*}
Let $L/k$ be any finite 
extension and let $C$ be an $L$-linear subspace of $L^n$ for some $n\geq 1$. Then 
$$C^*=\Rsupp(C)_L\,.$$
In particular we have $\, \dim C^*=\wt_R(C) ,\ $ so that $\dim C\leq \wt_R(C)$, and 
 $\dim (L{\bf c})^*=\wt_R({\bf c})$ for any ${\bf c}\in L^n$.
\end{prop}

\begin{proof}
We know from Proposition \ref{prop:rkbasic} that $\Rsupp(C)_L$ contains $C$. Since it is also extended from $k^n$, 
$\Rsupp(C)_L\,$ contains $C^*$. Conversely, $\Rsupp(C)_L$  is included in $\Rsupp(C^*)_L$,
and since the subspace $C^*$ is extended from $k^n$, we have $\Rsupp(C^*)_L=C^*$ (by Proposition
\ref{k-bas-Res-Rsupp}).
Whence $\Rsupp(C)_L$ is equal to $C^*$. Now use that 
 $\Rsupp(C)_L$ has dimension $\dim \Rsupp(C)=\wt_R(C)$, and that $C\subset C^*$. Finally we write
 our equality for $C=L{\bf c}$ and use statement 2 of \ref{prop:rkbasic}.
\end{proof}

\begin{rem}\label{CNdim}
The equality  $ \dim C^*=\wt_R(C)$ implies that the condition $\dim C^*\leq m$ is necessary for the existence of
 an element ${\bf c}\in C$  such that $\mathrm{Rsupp}(C)=\mathrm{Rsupp}({\bf c})$. 
Indeed,  if ${\bf c}\in L^n$ then its matrix $M({\bf c})$  has $m$ rows so that 
$\mathrm{dim}\,\mathrm{Rsupp}({\bf c})\leq m$.
In particular, the weaker condition 
$\dim C\leq m$ is 
necessary. Note that, whereas 
we have shown in Proposition \ref{Thm-cas r<m} that  in case $C$ is extended from $k^n$ 
this last condition is also sufficient, 
 in case $C$ is {\sl not} extended from $k^n$
and such an element 
${\bf c}$ 
exists, we deduce from 
our Proposition that $\dim C\leq m-1$.
\end{rem}

We give two corollaries to Proposition \ref{Rsupp-C*}. Note that the second one  
admits an elementary proof, using the Galois group, in  case the extension is Galois.

\begin{cor}\label{equiv-thm-<c>*}
Let $L/k$ be a finite  extension, let $n\geq 1$,  let $C$ be an $L$-linear subspace of $L^n$ and let ${\bf c}\in C$.
Then we have $\mathrm{Rsupp}(C)=\mathrm{Rsupp}({\bf c})$ if and only if  $C^*=(L{\bf c})^*$, that is, if and 
only if $C\subset (L{\bf c})^*$.
\end{cor}

\begin{proof}
Since $\mathrm{Rsupp}({\bf c})\subset\mathrm{Rsupp}({C})$ and $(L{\bf c})^*\subset C^*$, both equalities
are equivalent to the equality of the respective dimensions as $k$ 
and $L$-subspaces. But Proposition \ref{Rsupp-C*} shows that the $k$-dimension of $\Rsupp(C)$  is the same as the 
$L$-dimension of $C^*$, and similarly when replacing  $C$ with $L{\bf c}$. 

Finally, if ${\bf c}\in C$ and $C\subset (L{\bf c})^*$,  we deduce  that  $C^*=(L{\bf c})^*$ by the second remark \ref{C*Galois=C*}. 
\end{proof}

\begin{cor}\label{RsuppC=C*}
Let $L/k$ be a finite  extension, let $n\geq 1$ and let $C$ be an $L$-linear subspace of $L^n$.
Then $\mathrm{Rsupp}(C)=\mathrm{Rsupp}(C^*)$. In particular if  $L/k$ is separable, we have  $\Tr(C)=\Tr(C^*)$.
\end{cor}

\begin{proof} 
Proposition \ref{Rsupp-C*} implies that the $k$-subspaces  $\Rsupp(C)$ and $\Rsupp(C^*)$ have 
the  same dimension. Thus the obvious inclusion $\mathrm{Rsupp}(C)\subset\mathrm{Rsupp}(C^*)$ is an equality.

 If  $L/k$ is separable, Theorem \ref{thm:Rsupp-Tr}  tells us that
$\Rsupp$ and $\Tr$  coincide on $L$-linear subspaces of $L^n$, so 
the proof is complete.
\end{proof}

Finally, we slightly generalize Proposition \ref{Thm-cas r<m}, 
which similarly works even in case $m<n$.

\begin{lem}\label{thm-+v-k^n}
Let $L/k$ be a finite  extension of degree $m$, let $n\geq 1$, and  let $C$ be an $L$-linear subspace of $L^n$ such that
$\dim C^*\leq m$. Suppose $C=C_1\oplus L{\bf c}_2$, where ${\bf c}_2\in k^n$, and that there exists
${\bf c}_1\in C_1$ such that $\mathrm{Rsupp}(C_1)=\mathrm{Rsupp}({\bf c}_1)$. Then there exists
${\bf c}\in C$ such that $\mathrm{Rsupp}(C)=\mathrm{Rsupp}({\bf c})$.
\end{lem}

\begin{proof}
Let $\alpha_1,\ldots,\alpha_m$ denote a $k$-basis of $L$.
We know from Proposition \ref{prop:rkbasic} that $\mathrm{Rsupp}({\bf c}_2)=k{\bf c}_2$ and that
$\mathrm{Rsupp}(C)$ is thus the sum of $\mathrm{Rsupp}({\bf c}_1)$ and $k{\bf c}_2$. 
If ${\bf c}_2\in\mathrm{Rsupp}({\bf c}_1)$ we are done, so we assume the contrary. Then the 
sum is direct, so that taking dimensions we get by Proposition \ref{Rsupp-C*} $\dim \mathrm{Rsupp}(C)\leq m$ and 
 $wt_R({\bf c}_1)\leq m-1$. Thus  one row of the matrix $M({\bf c}_1)$, say the last one, 
is a linear combination of the  other ones. We change the vectors 
$\alpha_1,\ldots,\alpha_{m-1}$ of the basis of $L$ by adding to each of  them a vector $\mu_i\alpha_m$, where
 $\mu_i\in k$ is chosen so that,
 in the new basis of $L$, the last row of  the corresponding matrix  $M({\bf c}_1)$ is zero.
We set ${\bf c}={\bf c}_1+\alpha_m{\bf c}_2\in C$ so that we have $\mathrm{Rsupp}({\bf c})=
\mathrm{Rsupp}({\bf c}_1)+k{\bf c}_2=\mathrm{Rsupp}(C)$.
\end{proof}

\begin{cor}\label{thm+Res} 
Let $L/k$ be a finite  extension of degree $m$, let $n\geq 1$, and  let $C$ be an $L$-linear subspace of $L^n$ such that
$\dim C^*\leq m$. Suppose that $C_1$ is an $L$-linear subspace  such that $C=C_1+ \Res(C)_L\,$,  and that there exists
${\bf c}_1\in C_1$ such that $\mathrm{Rsupp}(C_1)=\mathrm{Rsupp}({\bf c}_1)$. Then there exists
${\bf c}\in C$ such that $\mathrm{Rsupp}(C)=\mathrm{Rsupp}({\bf c})$.
\end{cor}

\begin{proof}
Let ${\bf e}_1,\ldots,{\bf e}_t$ be vectors of $\Res(C)$ which complete a basis of $C_1$ to an $L$- basis of $C$.
We may repeatedly apply Lemma \ref{thm-+v-k^n}, taking successively $C_1^{(i)}=C_1+\sum_{j<i} L{\bf e}_j$.
${\bf c}_2={\bf e}_i$, $1\leq i\leq t$.
\end{proof}

\section{Generalized rank weights for finite extensions}\label{equiv-def}

We now set notation 
in order to discuss the various definitions of the  generalized rank weights
of a code and their equivalence.
In what follows, we let $k\subset L$ be a finite 
extension of degree $m$, and let $n\in\N^*$. We let $C$ 
be a linear subspace of $L^n$ and let $r$ be an integer such that $1\leq r\leq \dim C$.
Recall from Remark \ref{C*Galois=C*} that our definition \ref{C*} of $C^*$ generalizes the definition of Galois closure
used in \cite{Jurrius17}. 

\begin{notation} We will denote by  $\maxwt_R(C)$  the maximum value  $\max\limits_{{\bf c}\in C}
\wt_R({\bf c})$.                     
\end{notation}

Clearly we have  $\maxwt_R(C)\leq \wt_R(C)$, and equality holds   
if and only if there exists ${\bf c}\in C$ such that  $\mathrm{Rsupp}(C)=\mathrm{Rsupp}({\bf c})$.

We now present, in the setting of an arbitrary  finite 
extension, the various definitions which have been proposed for the $r$-th generalized rank weight
of $C$. We insist that the star  symbol 
used in definitions 2.~and 4.~refers to our definition \ref{C*}, so that 2.~and 4.~are indeed generalized
versions of the definitions considered in \cite{Jurrius17}.

\begin{defn}\label{4def-grw} 
\begin{enumerate}
\item $d_{R,r}(C)=\min\limits_{\stackrel{D\subset C} {\dim(D)=r}}\wt_R(D)$, following \cite[Definition 2.5]{Jurrius17} 
\item ${\cal M}_r(C)=\min\limits_{\stackrel{V\subset L^n,V=V^*} {\dim(C\cap V)\geq r}}\dim V$,   following  \cite[Definition 5]{Kurihara15}
\item   $O\!S_r(C)=\min\limits_{\stackrel{D\subset C} {\dim(D)=r}} \maxwt_R(D)$,  following \cite{Oggier12}
\item  $D_r(C)=\min\limits_{\stackrel{D\subset C} {\dim(D)=r}} \maxwt_R(D^*)$, following \cite{Ducoat15}. 
\end{enumerate}
\end{defn}

Jurrius-Pellikaan have shown in \cite [4.4 and 5.4]{Jurrius17}, that $d_{R,r}(C)={\cal M}_r(C)$
 whenever  $L/k$ is Galois.
They also proved in \cite[Theorem 5.8]{Jurrius17} that if $m\geq n$ and $L/k$ is cyclic, then ${\cal M}_r(C)=D_r(C)$.

Our main statement in this section is the following result.

\begin{thm}\label{equivdef} 
Let $L/k$ be a finite 
extension of degree $m$, let $n\geq 1$ be an integer such that $n\leq m$ and let $C$ be an $L$-linear subspace of $L^n.$ 
Then  
 all four definitions \ref{4def-grw} coincide.
\end{thm}

\begin{proof}
Since $m\geq n$, Theorem \ref{thm:geneJuPe} implies that $ \maxwt_R(D)= \wt_R(D)$, for all subspaces $D$ of $L^n$.
This gives  $d_{R,r}(C)=O\!S_r(C)$. 
To prove $O\!S_r(C)=D_r(C)$, 
we combine Theorem \ref{thm:geneJuPe} with the fact  that, 
 by Corollary \ref{RsuppC=C*},  $D$ and $D^*$  have the same rank support, and therefore 
$ \wt_R(D)= \wt_R(D^*)$. 

 Finally, the proof given in \cite [4.4 and 5.4]{Jurrius17}  that $d_{R,r}(C)={\cal M}_r(C)$
 whenever  $L/k$ is Galois relies both on the equality   $\dim D^*=\wt_R(D)$, which we have proved in \ref{Rsupp-C*}
 in the setting of arbitrary finite extensions, 
and on the properties of $D^*$ and $V^*$ stated in the second remark  \ref{C*Galois=C*}. Thus this  equality 
still holds in  our setting.
\end{proof}


\bibliographystyle{alpha}
\bibliography{General}

\end{document}